\documentclass[11pt]{article}

\usepackage[linesnumbered,vlined,ruled]{algorithm2e}
\usepackage{multirow,amsmath,amsfonts}

\usepackage{amsmath}
\usepackage{amssymb}
\usepackage{latexsym}
\usepackage{epsfig}
\usepackage{array}
\usepackage{color}
\usepackage{enumerate}
\usepackage{fullpage}
\usepackage[margin=1in]{geometry}

\newtheorem{definition}{Definition}

\newtheorem{theorem}{Theorem}
\newtheorem{lemma}{Lemma}

\newtheorem{corollary}{Corollary}

\newlength{\boxwidth}
\makeatletter
\DeclareRobustCommand{\qed}{%
  \ifmmode %
  \else \leavevmode\unskip\penalty9999 \hbox{}\nobreak\hfill
  \fi
  \quad\hbox{\qedsymbol}}
\newcommand{\openbox}{\leavevmode
  \hbox to.77778em{%
  \hfil\vrule \vbox to.675em{\hrule width.6em\vfil\hrule}%
  \vrule\hfil}}
\newcommand{\qedsymbol}{\openbox}

\newenvironment{proof}[1][\proofname]{\par \normalfont
  \topsep6\p@\@plus6\p@ \trivlist
  \item[\hskip\labelsep\bfseries\itshape #1.]\ignorespaces }{%
  \qed\endtrivlist }

\DeclareMathOperator{\supp}{Supp}

\newcommand{\proofname}{Proof}

\newcommand{\hide}[1]{}

\newcommand{\NE}{Nash equilibrium}
\newcommand{\NA}{Nash equilibria}
\newcommand{\CC}{communication complexity}
\newcommand{\ANE}{$\epsilon$-approximate Nash equilibrium}

\newcommand{\NNE}{-approximate Nash equilibrium}
\newcommand{\WSNE}{well-supported Nash equilibrium}
\newcommand{\player}{p}
\newcommand{\T}{{\rm T}}

\renewcommand{\vec}[1]{\mathbf{#1}}
\newcommand{\argmin}{\mathop{\rm argmin}}
\newcommand{\argmax}{\mathop{\rm argmax}}
\newcommand{\vc}[1]{\mathbf{#1}}
\newcommand{\bm}[1]{\mbox{\boldmath{$#1$}}}
\newcommand{\noitemsep}{\setlength{\itemsep}{-\parskip}}

\begin{document}

\title{On the Communication Complexity of Approximate Nash
Equilibria\footnote{A preliminary version of this paper appeared in the
Proceedings of the 5th SAGT. The first author was supported by EPSRC Grant
EP/G069239/1 ``Efficient Decentralised Approaches in Algorithmic Game Theory''}}

\author{Paul W. Goldberg$^1$, Arnoud Pastink$^2$\\ \\ 
$^1$ University of Liverpool\\
Dept.\ of Computer Science \\
Ashton Street, Liverpool L69 3BX, U.\,K.\\
\texttt{P.W.Goldberg@liverpool.ac.uk} \\ \\
$^2$ Utrecht University \\
Department of Information and Computing Science \\
P.O. Box 80089, 3508TB Utrecht, The Netherlands \\
\texttt{A.J.Pastink@uu.nl}
}


\maketitle

\def\sm{\setminus}
\def\C{{\vec c}}
\def\Cu{{c}}
\def\Ci{{c}}
\newcommand{\Co}[1] {c_{#1}}
\def\e{e}

\def\B{{\vec b}}
\def\Bu{{\vec b}}
\def\Bi{{b}}
\def\Bipp{{b''}}
\newcommand{\Bo}[1] {b_{#1}}
\newcommand{\Buo}[2] {b^{#1}_{#2}}

\def\mm{\mathcal{M}}
\def\me{\mathcal{E}}
\def\mf{\mathcal{F}}

\def\SI{\sigma} 
\def\SIGA{\alpha}
\def\SIGB{\beta}
\def\SIGC{\gamma}

\def\CSA{\C^{\SIGA}}
\def\CSB{\C^{\SIGB}}
\def\CSC{\C^{\SIGC}}

\def\BSA{\B^{\SIGA}}
\def\BSB{\B^{\SIGB}}
\def\BSC{\B^{\SIGC}}
\def\BS{\B^{\SI}}
\def\BSI{\Bi^{\SI}}
\def\BSS{\Buo{\SI}{S}}
\def\BSO{\Buo{\SI}}

\def\BMIN{\B^{\min}} 
\def\BMINi{\Bi^{\min}} 

\newcommand{\BMINo}[1] {\Buo{\min}{#1}}
\def\BMINp{\B'^{\min}} 
\def\BMINpi{\Bi'^{\min}} 
\def\BMINpo{\Buop{\min}} 

\def\m{k} 

\setlength{\tabcolsep}{.5em}

\begin{abstract}
We study the problem of computing approximate Nash equilibria of bimatrix games, in a setting
where players initially know their own payoffs but not the payoffs of the other player.
In order for a solution of reasonable quality to be found, some
amount of communication needs to take place between the players.
We are interested in algorithms where the communication is substantially less
than the contents of a payoff matrix, for example logarithmic in the size of the matrix.
When the communication is polylogarithmic in the number of strategies $n$,
we show how to obtain \ANE\ for $\epsilon$ approximately $0\cdotp438$, and for well-supported
approximate equilibria we obtain $\epsilon$ approximately $0\cdotp732$.
For one-way communication
we show that $\epsilon=\frac{1}{2}$ is achievable, but no constant improvement
over $\frac{1}{2}$ is possible, even with unlimited one-way communication.
For well-supported equilibria, no value of $\epsilon<1$ is achievable with
one-way communication.
When the players do not communicate at all, $\epsilon$-Nash equilibria
can be obtained for $\epsilon=\frac{3}{4}$, and we also give a lower bound of slightly
more than $\frac{1}{2}$ on the lowest constant $\epsilon$ achievable.
\end{abstract}

\section{Introduction}

Algorithmic game theory is concerned not just with properties of a solution concept,
but also how that solution can be obtained.
It is considered desirable that the outcome of a game should be ``easy to compute'',
which is typically formalised as polynomial-time computability, in the
algorithms community.
In that respect the PPAD-completeness results of~\cite{Daskalakis2006,Chen2006}
are interpreted as a ``complexity-theoretic critique'' of Nash equilibrium.
Following those results, a line of work addressed the problem of computing
$\epsilon$-Nash equilibrium, where $\epsilon>0$ is a parameter that bounds
a player's incentive to deviate, in a solution. Thus, $\epsilon$-Nash equilibrium
imposes a weaker constraint on how players are assumed to behave, and an exact
Nash equilibrium is obtained for $\epsilon=0$. The main open problem is to find out
what values of $\epsilon$ admit a polynomial-time algorithm. Below we
summarise some of the progress in this direction.

Beyond the existence of a fast algorithm, it is also desirable that a solution
should be obtained by a process that is simple and decentralised, since that is
likely to be a better model for how players in a game may eventually reach
a solution. In that respect, most of the known efficient algorithms for computing
$\epsilon$-Nash equilibria are not entirely satisfying. They take as input
the payoff matrices and output the approximate Nash equilibrium. If we try to
translate such an algorithm into real life, it would correspond to a process
where the players pass their payoffs to a central authority, which returns to
them some mixed strategies that have the ``low incentive to deviate'' guarantee.
In this paper we aim to model a setting where players perform individual
computations and exchange some limited information. We revisit the question
of what values of $\epsilon$ are achievable, subject to this restriction to
more ``realistic'' algorithms.

There are various ways in which one can try to model the notion of a
decentralised algorithm; here we consider a general approach that has previously
been studied in~\cite{Conitzer2004,Hart2010} in the context of computing exact
Nash equilibria. The players begin with knowledge of their own payoffs
but not the payoffs of the other players; this is often called an {\em uncoupled}
setting (see Section~\ref{sec:uncoupled} for an overview).
An algorithm involves communication
in addition to computation; to find a game-theoretic solution, a player
usually has to know something about the other players' matrices, but
hopefully not all of that information. We study the computation of $\epsilon$-Nash
equilibria in this setting, and the general topic is the trade-off between
the amount of communication that takes place, and the value of $\epsilon$ that
can be obtained. In uncoupled settings, there are natural dynamic processes
that converge to correlated equilibria, but the results are less positive
for exact Nash equilibria, so this paper can be seen as an investigation into
approximate Nash equilibrium as an alternative to correlated equilibrium,
as a solution concept.

\subsection{Definitions}\label{sec:defs}

We consider 2-player games, with a \textit{row player} and a
\textit{column player}, who both have $n$ \textit{pure strategies}. The
game $(R,C)$ is defined by two $n\times n$ \textit{payoff matrices},
$R$ for the row player, and $C$ for the column player. The pure
strategies for the row player are his rows and the pure strategies of
the column player are her columns. If the row player plays row $i$ and
the column player plays column $j$, the \textit{payoff} for the row
player is $R_{ij}$, and $C_{ij}$ for the column player.  For the row
player a \textit{mixed strategy} is a probability distribution $\vc{x}$
over the rows, and a mixed strategy for the column player is a
probability distribution $\vc{y}$ over the columns, where $\vc{x}$ and
$\vc{y}$ are column vectors and $(\vc{x},\vc{y})$ is a \textit{mixed
strategy profile}. The payoffs resulting from these mixed strategies
$\vc{x}$ and $\vc{y}$ are $\vc{x}^\T R\vc{y}$ for the row player and
$\vc{x}^\T C\vc{y}$ for the column player.

A \textit{\NE}\ is a pair of mixed strategies $(\vc{x}^*,\vc{y}^*)$
where neither player can get a higher payoff by playing another strategy
assuming the other player does not change his strategy. Because of the
linearity of a mixed strategy, the largest gain can be achieved by
defecting to a pure strategy. Let $\vc{e}_i$ be the vector with a 1 at
the $i$th position and a 0 at every other position. Thus a \NE\ $(\vc{x}^*,\vc{y}^*)$
satisfies
\[
\forall i = 1\cdots n ~~~ \vc{e}_i^\T R\vc{y}^*\leq(\vc{x}^*)^\T R\vc{y}^* ~~{\rm and}~~
(\vc{x}^*)^\T C\vc{e}_i\leq(\vc{x}^*)^\T C\vc{y}^*.
\]
We assume that the payoffs of $R$ and $C$ are between 0 and 1,
which can be achieved by rescaling.
An \textit{\ANE}\ (or, $\epsilon$-Nash equilibrium) is a strategy pair
$(\vc{x}^*,\vc{y}^*)$ such that each
player can gain at most $\epsilon$ by unilaterally deviating to a
different strategy.
Thus, it is $(\vc{x}^*,\vc{y}^*)$ satisfying
\[
\forall i = 1\cdots n ~~~ \vc{e}_i^\T R\vc{y}^*\leq(\vc{x}^*)^\T R\vc{y}^*+\epsilon ~~{\rm and}~~
(\vc{x}^*)^\T C\vc{e}_i\leq(\vc{x}^*)^\T C\vc{y}^*+\epsilon.
\]
We say that the {\em regret} of a player is the difference between his
payoff and the payoff of his best response.

The \textit{support} of a mixed strategy $\vc{x}$, denoted $\supp(\vc{x})$,
is the set of pure strategies that are played with non-zero probability by
$\vc{x}$.
An {\em approximate well-supported Nash equilibrium} strengthens
the requirements of a mixed Nash equilibrium. For a mixed strategy $\vc{y}$ of the
column player, a pure strategy $i \in [n]$ is an \emph{$\epsilon$-best response}
for the row player if, for all pure strategies $i' \in [n]$ we have:
$\vc{e}^T_iR\vc{y}\geq \vc{e}^T_{i'}R\vc{y}-\epsilon$.
We define $\epsilon$-best responses for the column player analogously.
A mixed strategy profile $(\vc{x}, \vc{y})$ is an
\emph{$\epsilon$-well-supported Nash equilibrium} ($\epsilon$-WSNE) if every
pure strategy in $\supp(\vc{x})$ is an $\epsilon$-best response against~$\vc{y}$, and
every pure strategy in $\supp(\vc{y})$ is an $\epsilon$-best response against~$\vc{x}$.

\begin{paragraph}{The communication model:}
Each player $\player\in\{r,c\}$ has an algorithm ${\cal A}_\player$ whose initial input
data is $\player$'s $n\times n$ payoff matrix. Communication proceeds in a number of rounds,
where in each round, each player may send a single bit of information to the
other player. During each round, each player may also carry out a
polynomial (in $n$) amount of computation.
(A natural variant of the model would omit the restriction to polynomial computation.
Indeed, our lower bounds on communication requirement do not depend on computational limits.)
At the end, each player $\player$ outputs a mixed
strategy $\vc{x}_\player$. We aim to design (pairs of) algorithms $({\cal A}_r,{\cal A}_c)$ that
output $\epsilon$-Nash strategy profiles $(\vc{x}_r,\vc{x}_c)$, and are economical with the
number of rounds of communication. This is similar to the {\em mixed
Nash equilibrium procedure} of~\cite{Hart2010}, here applied to approximate
rather than exact equilibria.

Notice that given $\Theta(n^2)$ rounds of communication, we can apply any
centralised algorithm ${\cal A}$ by getting (say) the row player to pass
additive approximations of all his payoffs to the
column player, who applies ${\cal A}$ and passes to the row player the mixed strategy
obtained by ${\cal A}$ for the row player.
(The quality of the $\epsilon$-Nash equilibrium is proportional to the quality of
of the additive approximations used.) For this reason we focus on algorithms
with many fewer rounds, and we obtain results for logarithmic or polylogarithmic
(in $n$) rounds.

We also consider a restriction to {\em one-way communication}, where one player may send
but not receive information.
\end{paragraph}

\subsection{Related Work}

We start by reviewing some algorithms that we adapt to the communication-bounded
setting. Then we review the background work on communication complexity,
and related work in computing Nash equilibria, including learning of
equilibria in {\em uncoupled} settings.

\subsubsection{Algorithms for approximate equilibria}\label{sec:approxalgos}

In recent years a number of algorithms~\cite{KPS,Daskalakis2009,DMP2,Bosse2007,Tsaknakis2007}
have been developed that
compute (in polynomial time) $\epsilon$-Nash equilibria for various values of $\epsilon$.
Of these, Tsaknakis and Spirakis~\cite{Tsaknakis2007} obtain the best (smallest)
value of $\epsilon$, of approximately $0\cdotp3393$.
The more demanding criterion of {\em well-supported} $\epsilon$-Nash equilibrium,
disallows a player from allocating positive probability to any pure strategy
whose payoff is more than $\epsilon$ worse than the best response.
Progress on polynomial-time algorithms for this
solution concept has been more limited; at this time the lowest $\epsilon$ that
can be guaranteed by a polynomial-time algorithm is only slightly less
than $\frac{2}{3}$~\cite{FGSS12}, obtained via a modification of a $\frac{2}{3}$-approximation
algorithm of Kontogiannis and Spirakis~\cite{KS10}.
Prior to that, \cite{Daskalakis2009} gave a $\frac{5}{6}$-approximation algorithm,
that is contingent on a graph-theoretic conjecture.
In this context, our $0\cdotp732$-approximation algorithm substantially
improves on the result of~\cite{Daskalakis2009}, both in terms of approximation
quality and a more demanding model (communication-bounded algorithms).
However, we do not know how to obtain the better approximation quality
of~\cite{KS10,FGSS12} in the communication-bounded setting.
Next we discuss two of the earlier algorithms in the literature whose ideas we use here.

\begin{paragraph}{DMP-algorithm:}
The DMP-algorithm~\cite{Daskalakis2009} works as follows to achieve a $\frac{1}{2}$\NNE.
The algorithm picks a arbitrary row for the row player, say row $i$.
Let $j\in \argmax_{j'} C_{ij'}$. Let $k\in \argmax_{k'} R_{k'j}$.
So $j$ is a pure-strategy best response for the
column player to row $i$ and $k$ is a best response strategy for the row
player to column $j$. The strategy pair $(\vc{x}^*,\vc{y}^*)$ will now
be $\vc{x}^*=\frac{1}{2}\vc{e}_i+\frac{1}{2}\vc{e}_k$ and
$\vc{y}^*=\vc{e}_j$. With this strategy pair the row player plays a best
response with probability $\frac{1}{2}$ to a pure strategy of the column
player and the column player has a pure strategy that is with
probability $\frac{1}{2}$ a best response.

The DMP-algorithm is well-adapted to the limited-communication setting.
Suppose the row player uses $i=1$ as his initial choice of row.
The column player needs to tell the row player her value of $j$, a
communication of $O(\log n)$ bits. No further communication is needed.
Notice moreover that the communication is all one-way; the row player
does not need to tell the column player anything.
\end{paragraph}

Subsequent algorithms for computing $\epsilon$-Nash equilibria cannot
so easily be adapted to a limited-communication setting, but we can use
some of the ideas they develop, to obtain values of $\epsilon$ below $\frac{1}{2}$
in this setting.

\begin{paragraph}{An algorithm of Bosse et al.~\cite{Bosse2007}:}
The algorithm presented in \cite{Bosse2007} can be seen as a
modification of the DMP-algorithm and achieves a $0\cdotp38197$\NNE. Instead of
a player allocating some probability to some arbitrary pure strategy,
the algorithm starts with the row player allocating some probability to
the row-player strategy $\vc{x}$ belonging to the \NE\ of the zero-sum game $(R-C,C-R)$.
In solving the zero-sum game efficiently we apply the connection of zero-sum games
with linear programming~\cite{Neumann28,Dantzig63,Karmarkar84}.
If the (mixed) strategy profile $(\vc{x},\vc{y})$ that constitutes
a \NE\ of $(R-C,C-R)$ gives a $0\cdotp38197$\NNE\ for $(R,C)$, this solution is used.
Otherwise, the column player plays a best response $\vc{e}_j$ to $\vc{x}$
and the row player plays a mixture of $\vc{x}$ and $\vc{e}_k$, where
$\vc{e}_k$ is a best response to the strategy $\vc{e}_j$ of the column
player. (\cite{Bosse2007} goes on to improve the worst-case performance
to a $0\cdotp36395$\NNE.)

Notice that this algorithm cannot be adapted in a straightforward way to
our communication-bounded setup, since it requires a computation using
knowledge of both matrices. The starting-point of our algorithms of Section~\ref{sec:commbounded}
is the players separately solving $(R,-R)$ and $(-C,C)$.
\end{paragraph}

\subsubsection{Communication Complexity}

The ``classical'' setting of \CC\ is based on the model introduced by
Yao in \cite{Yao1979}. We will follow the representation in
\cite{Kushilevitz1997}. We have two agents\footnote{We use agents
instead of players to avoid confusion, the communication does not have to
be between the players of the game.}, one holding an input
$\vc{x}\in\{0,1\}^n$ and the other holding an input
$\vc{y}\in\{0,1\}^n$. The objective is to compute $f(\vc{x},\vc{y})\in
\{0,1\}$, a joint function of their inputs. The computation of
$f(\vc{x},\vc{y})$ is done via a communication protocol $\mathcal{P}$.
During the execution of the protocol, the agents send messages to
each other. While the protocol has not terminated, the protocol
specifies what message the sender should send next, based on the input
of the protocol and the communication so far. If the protocol
terminates, it will output the value $f(\vc{x},\vc{y})$. A communication
protocol $\mathcal{P}$ computes $f$ if for every input pair
$(\vc{x},\vc{y})\in\{0,1\}^n\times\{0,1\}^n$, it terminates
with the value $f(\vc{x},\vc{y})$ as output.

The \CC\ of a communication protocol $\mathcal{P}$ for computing
$f(\vc{x},\vc{y})$ is the number of bits sent during the execution of 
$\mathcal{P}$, which we denote by
$CC(\mathcal{P},f,\vc{x},\vc{y})$. The \CC\ of a protocol $\mathcal{P}$
for a function $f$ is defined as the worst case \CC\ over all possible
inputs for $(\vc{x},\vc{y})\in\{0,1\}^n\times\{0,1\}^n$, which we denote
by $CC(\mathcal{P},f)$:
$$CC(\mathcal{P},f)=\max_{(\vc{x},\vc{y})\in\{0,1\}^n\times\{0,1\}^n}CC(
\mathcal{P},f,\vc{x},\vc{y})$$ The \CC\ of a function $f$ is the minimum
over all possible protocols:
$$CC(f)=\min_{\mathcal{P}}CC(\mathcal{P},f)$$

\subsubsection{Existing results on \CC\ of \NA}

There are a few results concerning the \CC\ of \NA.
Conitzer and Sandholm~\cite{Conitzer2004} show a lower bound on the \CC\ for
2-player games of finding a pure \NE\ of $\Omega(n^2)$, where $n$ is the
number of pure strategies for each player. They also show a simple
algorithm that finds a pure \NE\ (if it exists) in $O(n^2)$. They do not
extend their analysis to mixed \NA; their focus is on searching for a
pure \NE\ (if one exists), in contrast with the existence of a mixed \NE,
which is guaranteed~\cite{Nash1951}. For unrestricted bimatrix games, it can
be seen that the communication complexity of finding an exact equilibrium
is $\Omega(n^2)$\footnote{Consider a game where there is a unique, fully-mixed
Nash equilibrium. If the payoffs are perturbed slightly, the resulting
equilibrium, for (say) the row player, will be affected in a non-trivial
way by all the perturbations of the column player's payoffs. This immediately
results in the requirement of $\Omega(n^2)$ communication.}. That observation
leads to the question addressed here, of whether approximate equilibria have
lower communication complexity.

Also related to this paper,
Hart and Mansour~\cite{Hart2010} study the \CC\ of uncoupled equilibrium procedures,
(discussed in more detail below in Section~\ref{sec:uncoupled})
in the context of multiplayer, binary action games.
The emphasis is on lower bounds on the communication requirement.
Analogously to the $\Omega(n^2)$ communication needed for pure or mixed
\NE\ that we noted above, they obtain a lower bound of $\Omega(2^s)$
(where $s$ is the number of players) on
the communication needed to find an exact mixed equilibrium, or determine
the existence of a pure one. (Note that in their setting, each player has
a payoff matrix of size $2^s$, so that essentially all the payoffs may need
to be communicated.) On the other hand, they obtain a polynomial upper
bound on the communication required to find a {\em correlated equilibrium},
discussed further below.
Their methods do not seem to be applicable in an obvious way to
approximate equilibria. For example, the lower bound for computing a mixed
equilibrium involves a game whose solution requires probabilities having
exponentially large descriptions, which
would not be needed in the context of approximate equilibria.

\subsubsection{Uncoupled Learning of Equilibria}\label{sec:uncoupled}

An extensive literature studies {\em uncoupled} procedures for finding
game-theoretic solutions. The terminology ``uncoupled'' is introduced in~\cite{HMC2003};
it refers to settings where each player knows his own (but not the others')
utility function. Then, there is a sequence of rounds (a.k.a. time steps, or
periods), in which each player plays a strategy, and receives the payoff
resulting from the entire strategy profile. Our setting of communication
complexity is related to this, in that each player can use his choice of
action (in a round) to transmit information. The main difference is that
here, we do not assume a ``rational'' choice of action where a player tries
to maintain his payoff over time by predicting the choices of his opponents.
In our set-up, player communicate some information over a (hopefully short)
sequence of rounds, and afterwards promise to use certain mixed strategies.
Our interest is in both upper and lower bounds on the required length of
the sequence. As noted in Conitzer and Sandholm~\cite{Conitzer2004},
lower-bound type results generally ignore strategic considerations, which
perhaps helps to justify our own inattention to rationality in this paper.

In the context of uncoupled search for Nash equilibrium, Hart and Mas-Colell~\cite{HMC2003}
show that when players do not remember the history of play, it may be
impossible to reach Nash equilibrium. Note that the obstacle is informational
rather than due to rationality of the players.
A subsequent paper~\cite{HMC2006} analyses how much of the history of play
needs to be recalled by the players. In the case of mixed (approximate) Nash
equilibria, the approach is to test many probability distributions is a search
for one that constitutes an approximate equilibrium; a large number of rounds
is required to achieve this.
Foster and Young~\cite{Foster2006} show how this can be achieved in a ``radically
uncoupled'' setup, where a player does not directly observe the opponents'
behaviour, but observed it indirectly via the payoffs he obtains. Again, a
very large number of rounds are required to find an approximate equilibrium.
Daskalakis et al.~\cite{DFPPV} study negative results, namely failure to
converge to Nash equilibrium, for standard multiplicative weights
update algorithms, in the context of bimatrix games. Their results consider three
variants of uncoupled dynamics.

There are more natural learning algorithms that converge (in various senses)
to the weaker solution concept of correlated equilibrium (e.g. Foster and
Vohra~\cite{FV1997}, Hart and Mas-Colell~\cite{HMC2000}).
When we relax our objective from approximate Nash equilibrium to approximate
correlated equilibrium, then learning can take place with a sublinear
number of rounds, from a straightforward
application of no-regret learning algorithms. The idea is applied in
Theorem~30 of~\cite{Hart2010}. In particular, we equip each
player\footnote{Indeed, there may be any number of players, not just 2.} with
a no-regret algorithm, and suppose that at each round it duly selects (and outputs)
a pure strategy, which requires $\log(n)$ bits to output.
Indeed, Theorem~17 of~\cite{Hart2010} shows how {\em exact} correlated equilibrium
may be found in a polynomial number of rounds.

Foster and Young~\cite{Foster2006} point out as motivation for uncoupled learning rules, 
that uncoupledness prevents a learning rule from behaving like a centralised
algorithm and just constituting a theory of equilibrium selection.
In this paper we similarly avoid the possibility of implementing a centralised
algorithm, though restricting to a sublinear number of rounds of communication,
so that it is impossible for one player to reveal all (or even a large fraction)
of his payoffs to the other player.

\subsection{Overview of our results}

For general $n\times n$ games we show the following bounds on the obtainable quality
of an approximate \NE\ if we fix the amount of communication allowed.
We start by considering a version where no communication is allowed.
Theorem~\ref{thm:UBnocomm} gives a simple way to find a $\frac{3}{4}$-Nash
equilibrium, in this setting.
Theorem~\ref{thm:halfplus} identifies a corresponding lower bound of slightly
more than $\frac{1}{2}$.
For one-way communication we exhibit (Theorem~\ref{thm:oneway})
a lower bound of $\frac{1}{2}-o(\frac{1}{\sqrt{n}})$.
The DMP-algorithm can be implemented as an
algorithm with one-way communication and gives a $\frac{1}{2}$\NNE. Therefore the
constant $\frac{1}{2}$ in the
lower bound of Theorem~\ref{thm:oneway} is tight, in this context.
In Section~\ref{sec:0.438} we show how to compute a $0\cdotp438$-Nash equilibrium
using polylogarithmic communication. In Section~\ref{sec:conclusions} we
discuss the significance of the results, along with open problems.

\section{Approximate Nash Equilibria with no Communication}

The simplest way to restrict communication is to disallow it entirely.\footnote{This
is to some extent inspired by earlier work
of the first author~\cite{Goldberg2006} that studied an approach to
pattern classification in which the set of observations of each class must be
processed by an algorithm that proceeds independently of the corresponding
algorithms that receive members of the other classes.}
That means that for each player $\player\in\{r,c\}$,
we must find a function $f_\player$ from $\player$'s payoff matrix to a mixed strategy,
such that for all pairs of matrices $(R,C)$, we have that $(f_r(R),f_c(C))$
is an $\epsilon$-Nash equilibrium.
In this section we show that the achievable value of $\epsilon$ lies somewhere
between $0\cdotp501$ and $\frac{3}{4}$. The $\frac{3}{4}$ upper bound is achieved
via a simple algorithm (differing from the $\frac{3}{4}$-approximation algorithm
of~\cite{KPS}, in terms of the solution it finds).
Theorem~\ref{thm:halfplus} presents the lower bounds of $0\cdotp501$.

Theorem~\ref{thm:oneway} in Section~\ref{sec:oneway} furnishes a lower bound of
$\frac{1}{2}$, even when one-way communication is permitted, and has a simpler proof
(the proof is similar to Case 2 in the proof of Theorem~\ref{thm:halfplus}).
This raises the question: why bother to include a complicated proof
(specific to the communication-free setting) whose result
is only a small improvement (over the one-way communication setting)?
The reason is that we rule out the possibility that
$\frac{1}{2}$ is in fact the answer, and as we discuss in the conclusions
(Section~\ref{sec:conclusions}), $\frac{1}{2}$ seems to arise frequently as a
barrier to progress in the study of algorithms for approximate Nash equilibria, so
it is informative to rule out that possibility.
Our lower bound of $0\cdotp501$ could be increased slightly by tweaking the parameters
of the proof, but we believe that the resulting progress would be incremental.

\begin{theorem}\label{thm:UBnocomm}
It is possible to guarantee a $\frac{3}{4}$\NNE,
even if there is no communication between the players.
\end{theorem}

\begin{proof}
Each player allocates probability $\frac{1}{2}$ to his first pure strategy, and $\frac{1}{2}$
to his best response to the other player's first pure strategy.
In detail, let $i\in \arg \max_{i'}R_{i'1}$ and let $j\in \arg \max_{j'}C_{1j'}$.
The approximate \NE\ will be $f_r(R)=\frac{1}{2}\vc{e}_1+\frac{1}{2}\vc{e}_i$ and
$f_c(C)=\frac{1}{2}\vc{e}_1+\frac{1}{2}\vc{e}_j$.

Let $i'$ be a best pure strategy response of the row player to $f_c(C)$. 
Then his incentive to deviate is
\[
\left(\frac{1}{2}R_{i'1}+\frac{1}{2}R_{i'j}\right)-
\left(\frac{1}{4}R_{11}+\frac{1}{4}R_{1j}+ \frac{1}{4}R_{i1}+\frac{1}{4}R_{ij}\right)
\]
\[
\leq
\left(\frac{1}{4}R_{i'1}+\frac{1}{2}R_{i'j}\right)-
\left(\frac{1}{4}R_{11}+\frac{1}{4}R_{1j}+\frac{1}{4}R_{ij}\right)
\leq
\frac{1}{4}R_{i'1}+\frac{1}{2}R_{i'j} \leq \frac{1}{4}+\frac{1}{2}=\frac{3}{4}
\]
where the first inequality holds because $i$ was a best response to column
1 (so $R_{i1}\geq R_{i'1}$) and the next inequalities hold because payoffs lie
in $[0,1]$. The same kind of argument holds for the column player. This proves the
theorem.
\end{proof}

The following lemma provides a construction that is used in
Theorems~\ref{thm:halfplus} and~\ref{thm:oneway}.

\begin{definition}\label{def:matrix}
Let $M_{n}$ be a matrix with $n$ columns and  ${n\choose k}$ rows, where
$k=\lfloor \sqrt{n}\rfloor$ and a row consists
of $k$ 1's and $(n-k)$ 0's. Every row is distinct, so the
${n\choose k}$ rows are all the possible sequences with $k$ 1's in a row of length $n$.
\end{definition}

\begin{lemma}\label{lem:key}
Suppose we have a bimatrix game where the row player's payoff matrix is $M_{n}$ (as in
Definition~\ref{def:matrix}). Let $\vc{x}$ be a mixed strategy for the row player.
Then, there exists a column of $M_{n}$ such that if the column player uses any
mixed strategy $\vc{y}$ that allocates probability $p$ to that column, then the row player's
regret is at least $p-O(1/\sqrt{n})$.
\end{lemma}

\begin{proof}
The rows of $M_n$ contain 1's in a fraction $\frac{k}{n}$ of their entries.
By symmetry, so do the columns, thus every column contains 
$\frac{k}{n}\cdot{n\choose k}$ 1's and
$(1-\frac{k}{n})\cdot{n\choose k}$ 0's (recall $k=\lfloor\sqrt{n}\rfloor$).

${\vc x}$ assigns a probability to each row of $M_n$.
Define an unnormalised probability distribution $\Phi$ over the columns as follows.
Let $\Phi(j)$ be the probability
that a 1 will be in column $j$ of $M_n$, given a row sampled from $\vc{x}$.
Note that $\Phi(j)\leq 1$, with equality when every
row that is played with positive probability has a 1 in column $j$.
Because every row contains $k$ 1's, the sum of over all values $j$ will sum to $k$:
$\sum_{j=1}^{n} \Phi(j)=k$.

We define column $m$ to be one with a lowest value of $\Phi$: $m \in \argmin_j \Phi(j)$.
We choose $m$ to be the special column in the statement of the Lemma, and
we suppose that the column player allocates probability $p$ to $m$.

Since the sum over all values $\Phi(j)$ is $k$ and there are $n$ columns,
this means that $\Phi(m)\leq\frac{k}{n}$.
When column $m$ is played (and we assume it is played with probability $p$) 
it gives the row player a payoff of 0 with a probability of at least $1-\frac{k}{n}$.

We now consider the row player's strategy ${\vc x}$ and construct an improved
response $\vc{x}^*$ as follows.
$\vc{x}^*$ will differ from $\vc{x}$ in the following way. For every row $i$ we
see if its $m$-th entry is a 1. If this is the case,
we do not change anything. If instead its $m$-th entry is a 0, we do
the following: look at the entries where there is a 1 in row $i$.
Of all the entries where there is a 1, we select the one to which
the column player's distribution $\vc{y}$ gives the lowest probability, say entry $a$.
(i.e. choose column $a\in \arg\min_{j~:~M_n[i,j]=1} \vc{y}[j]$.)
Now we move all the probability allocated to row $i$ by $\vc{x}$,
to the row of $M_n$ that instead has a 0 in entry $a$ and a 1 in entry $m$,
and is otherwise the same as $i$.

The probability on entry $a$ is defined as the smallest among all the entries
where row $i$ has a 1. We can bound the probability that is allocated to
this entry by distribution $\vc{y}$. A probability at least $p$
is given to column $m$, so a probability of $1-p$ can be
distributed over the remaining columns. The column containing entry
$a$ has the smallest probability among at least $k$ columns, so the
probability given to column $a$ is at most $\frac{1-p}{k}$. 

The result of this construction of $\vc{x}^*$ from $\vc{x}$ is that every
row that is played with positive probability by $\vc{x}^*$ will have a 1 in the $m$-th entry.
There is a probability at least $(1-\frac{k}{n})$ that a row sampled from
$\vc{x}$ does not have a 1 in the $m$-th entry. This means that the increase
in payoff from replacing $\vc{x}$ with $\vc{x}^*$ is at least
\[
\left(1-\frac{k}{n}\right)\cdot p -\left(1-\frac{k}{n}\right)\cdot
\frac{1-p}{k} = \left(1-\frac{k}{n}\right)\cdot\left( p -\frac{1-p}{k}\right)
\geq p - \frac{1}{k}.
\]
Noting that $k=\lfloor \sqrt{n}\rfloor$ gives us the desired result.
\end{proof}

We use the following technical extension in the proof of Theorem~\ref{thm:halfplus}.
It is a straighforward corollary of Lemma~\ref{lem:key}.

\begin{corollary}\label{cor:key}
Suppose we have a bimatrix game where the row player's payoff matrix $R$ has $M_{n}$
(as in Definition~\ref{def:matrix}) as a submatrix. Suppose furthermore that all rows
of $R$ that do not intersect $M_n$ pay the row player zero, and for all rows of $R$
that intersect $M_n$, all columns that are not columns of $M_n$ pay the row player 1.

Let $\vc{x}$ be any mixed strategy for the row player that
allocates probability at least $p_r$ to rows that do not intersect $M_n$.
Let $\vc{y}$ be a mixed strategy for the column player, that allocates at least $p_c$ to columns that
do not intersect $M_n$, but allocates probability $p$ to some column $\ell$ intersecting $M_n$.
Then, there exists a choice of column $\ell$ such that the row player's
regret is at least $p - O(1/\sqrt{n}) + p_rp_c$.
\end{corollary}

\begin{proof}
Suppose $\vc{x}$ is modified as follows.
For rows that intersect $M_n$, modify their probabilities according to Lemma~\ref{lem:key}.
For other rows, set their probability to 0, and transfer their probability to
an arbitrary row that has payoff 1 when column $\ell$ is played.

This change increases by $p-O(1/\sqrt{k})$, the payoffs to the row player resulting
from the column player playing columns containing $M_n$. Note
that this gain is not conditioned on the column player playing columns intersecting
$M_n$; it is an absolute gain. In more detail, for rows intersecting $M_n$,
a fraction $1-\frac{k}{n}$ of them (w.r.t. probability measure $\vc{x}$)
have their payoff raised by at least $p-\frac{1-p}{k}$. For rows not intersecting
$M_n$, their payoffs are raised by at least $p$.

There is an additional increase to the row player's payoff due to the transfer of
probability from rows not intersecting $M_n$ to rows intersecting $M_n$, in the event
that the column player plays a column not containing $M_n$. In this case the
payoffs increase from 0 to 1, resulting in an additional payoff to the row
player of (at least) $p_r.p_c$.
\end{proof}

In the communication-free setting, each player $p$ computes a function
$f_p$ from his payoff matrix to a mixed strategy.
We will first introduce a ``commitment measure'' $\tau^p$ that measures the
variability of mixed strategies that may be selected by $p$, i.e.
the image of the set of all payoff matrices under $f_p$.

The variation distance between two probability distributions $\vc{x}$ and $\vc{x}'$
over $[n]$, is half the sum of all positive differences between the two distributions, i.e.
\[
d(\vc{x},\vc{x}')=\sum_{i=1}^n \frac{1}{2}\bigl|\vc{x}[i]-\vc{x}'[i]\bigr|.
\]
For $n\times n$ games, let $\Omega^r_n$ denote the set of strategies the row player may
use (i.e. the image of $f_r$) and $\Omega^c_n$ the set of strategies the column player may
use. For each player we define his ``centre strategy''. For the row player
the strategy $\vc{c}^r_n$ is the probability distribution such
that the maximum distance between $\vc{c}^r_n$ and any strategy
$\bm{\omega}\in\Omega^r_n$ is minimised.
\[
\vc{c}^r_n = \argmin_{\vc{c}} \sup_{\bm{\omega}\in\Omega^r_n} d(\vc{c},\bm{\omega})
\]
The centre distribution $\vc{c}^c_n$ of the column player is defined in a similar way.
The {\em commitment} $\tau^r_n$ of the row player is defined as
\[
\tau^r_n = 1-\sup_{\bm{\omega}\in\Omega^r_n} d(\vc{c}^r_n,\bm{\omega})
\]
The commitment $\tau^c_n$ of the column player is defined similarly.
This commitment measure $\tau^r_n$ will be a value in $[0,1]$ that indicates the variability
of strategies a player may use, and is high when the player always plays strategies
that are close to some ``central'' strategy.

\begin{theorem}\label{thm:halfplus}
For bimatrix games with payoffs in the range $[0,1]$,
if each player independently computes a mixed
strategy based on his own payoff matrix, then it is impossible to guarantee an
$\epsilon$-\NNE\ for $\epsilon<0\cdotp501$.
\end{theorem}

\begin{proof}
The proof will be a case analysis on commitment.
In the proof, our analysis is with respect to an
arbitrary fixed value of $n$, so we drop the subscript $n$ from
the commitment values $\tau^r_n$ and $\tau^c_n$, also the centre probability
vectors $\vc{c}^r_n$ and $\vc{c}^c_n$.
We will show that for all $n$, the regret of a player is at least $0\cdotp501$.
We identify two cases:
\begin{enumerate}\noitemsep
\item\label{case2} A player has a low commitment: $\tau^r\leq 0\cdotp05$ or $\tau^c\leq 0\cdotp05$
\item\label{case3} Neither player has a low commitment: $0\cdotp05<\tau^r$ and $0\cdotp05<\tau^c$
\end{enumerate}

\bigskip\noindent
{\bf Case~\ref{case2}: A player has low commitment}

\smallskip\noindent
Assume the column player has low commitment, thus $\tau^c\leq0\cdotp05$.
We use this low commitment to identify a set of strategies
that are quite far apart from each other, under variation distance.

For the column player, take an arbitrary strategy $\vc{s}_1 \in
\Omega^c$. Because $\tau^c\leq0\cdotp05$, there must be
some strategy $\vc{s}_2$ with $d(\vc{s}_1,\vc{s}_2)\geq 0\cdotp95$, otherwise
$\vc{s}_1$ could be the centre strategy $\vc{c}$ with $\tau^c\geq0\cdotp05$.

Now consider the strategy $\vc{s}_{12}=\frac{\vc{s}_1+\vc{s}_2}{2}$,
thus $d(\vc{s}_{12},\vc{s}_1)=d(\vc{s}_{12},\vc{s}_2)=\frac{1}{2}d(\vc{s}_1,\vc{s}_2)\leq\frac{1}{2}$.
For this strategy not to be a centre strategy $\vc{c}$ contradicting $\tau^c\leq 0\cdotp05$,
there must be some strategy $\vc{s}_3\in\Omega^c$ with
$d(\vc{s}_{12},\vc{s}_3)\geq0\cdotp95$.  Because $\vc{s}_1$ constitutes half of the
strategy $\vc{s}_{12}$, it holds that $d(\vc{s}_1,\vc{s}_3)\geq0\cdotp90$ and
similarly $d(\vc{s}_2,\vc{s}_3)\geq0\cdotp90$. We have
\[
d(\vc{s}_1,\vc{s}_2)    \geq 0\cdotp95;~~
d(\vc{s}_1,\vc{s}_3)    \geq 0\cdotp90;~~
d(\vc{s}_2,\vc{s}_3)    \geq 0\cdotp90;~~
d(\vc{s}_{12},\vc{s}_3) \geq 0\cdotp95.
\]
The next step is to construct a $n\times n$ payoff matrix $R$ of the row player.
Only the first 3 rows of $R$ will contain non-zero entries.
The construction of rows 1,2,3 will be such that for $i,j\in\{1,2,3\}$,
row $i$ is a best response to $s_i$ and a poor response to $s_j$ $(j\neq i)$.

For every column $j$ of $R$ determine the maximum of $\vc{s}_1[j]$,
$\vc{s}_2[j]$ and $\vc{s}_3[j]$. If $\vc{s}_1[j]$ is the largest,
$R_{1j}=1$ and  $R_{2j} = R_{3j}=0$. If $\vc{s}_2[j]$ is the largest,
$R_{2j}=1$ and  $R_{1j} = R_{3j}=0$. If $\vc{s}_3[j]$ is the largest,
$R_{3j}=1$ and  $R_{1j} = R_{2j}=0$. In case of a tie in the
comparison of $\vc{s}_1[j]$, $\vc{s}_2[j]$ and $\vc{s}_3[j]$, all the
entries corresponding to the tie get a 1.

Consider columns $i$ for which $R_{2i}=1$, so that $\vc{s}_2[i]>\vc{s}_1[i]$.
The total probability assigned by $\vc{s}_1$ to these columns is bounded by $0\cdotp05$.
If the probability on these columns was higher than $0\cdotp05$,
it would follow that $d(\vc{s}_1,\vc{s}_2)<0\cdotp95$.
Similarly we can bound the probability assigned by $\vc{s}_1$ to columns $i$
such that $R_{3i}=1$.
Since $d(\vc{s}_1,\vc{s}_3)\geq0\cdotp9$ this probability at most $0\cdotp1$.
From these observations, we have that at most $0\cdotp15$ of the probability
distribution $\vc{s}_1$ is allocated to columns that could give a payoff of 0 for row 1.
Since each column of $R$ contains at least one 1,
the remaining $0\cdotp85$ probability of $\vc{s}_1$ will be allocated to columns
that have a 1 in the corresponding entry of row 1.
The payoff for row 1 if
the column player plays $\vc{s}_1$ is therefore at least $0\cdotp85$. We can
use a similar argument to claim that when the column player plays
$\vc{s}_2$, the row player can get a payoff of at least $0\cdotp85$ by playing
pure strategy row 2, and at most $0\cdotp05$ for row 2, and at most $0\cdotp1$ for row 3.

For row 3 we use $d(\vc{s}_{12},\vc{s}_3)\geq0\cdotp95$.
Consider columns $i$ for which $R_{3i}=0$, so that either $R_{1i}=1$ or $R_{2i}=1$.
A column $i$ having this property, contributes $\geq\frac{1}{2}\vc{s}_3[i]$ to the
overlap between $\vc{s}_3$ and $\vc{s}_{12}$.
Indeed, if both $R_{1i}=1$ and $R_{2i}=1$, it contributes $\vc{s}_3[i]$ to the overlap.
So we can deduce that with respect to columns selected using $\vc{s}_3$,
$\Pr({\rm row}~1~{\rm pays}~1) + \Pr({\rm row}~2~{\rm pays}~1) \leq 0\cdotp1$.
Again, since each column of $R$ contains at least one 1,
the remaining $0\cdotp9$ probability of $\vc{s}_3$ will be allocated to columns
that have a 1 in the corresponding entry of row 3.
The payoff for row 3 if the column player plays $\vc{s}_3$
is therefore at least $0\cdotp9$, while the payoffs to rows 1 and 2 sum to at most $0\cdotp1$.
To summarise:
\begin{itemize}\noitemsep
\item If the column player plays $\vc{s}_1$, the row player gets a payoff
of at least $0\cdotp85$ by playing row 1. Playing row 2 would give him a payoff
of at most $0\cdotp05$ and playing row 3 a payoff of at most $0\cdotp1$. 
\item If the column player plays $\vc{s}_2$, the row player gets a payoff
of at least $0\cdotp85$ by playing row 2. Playing row 1 would give him a payoff
of at most $0\cdotp05$ and playing row 3 a payoff of at most $0\cdotp1$.
\item If the column player plays $\vc{s}_3$, the row player gets a payoff
of at least $0\cdotp9$ by playing row 3. Playing row 2 would give him a payoff
of at most $0\cdotp1$ and playing row 3 a payoff of at most $0\cdotp1$.
Moreover, the sum of payoffs of row 1 and row 2 is at most $0\cdotp1$.
\end{itemize}
Given the row player's strategy, let $(r_1,r_2,r_3)$ be the probabilities with
which he plays rows 1,2,3. Assume $r_1\leq
r_2,r_3$, so $r_1\leq\frac{1}{3}$ and suppose the column player plays strategy $\vc{s}_1$.
The best response strategy $(1,0,0)$ has a payoff of
$a\in[0\cdotp85,1]$. Because row 1 clearly gives the highest payoff, the
regret is minimised when this row is played with as much probability as
possible, so $r_1=\frac{1}{3}$. Because the probability on row 1 was defined
as the lowest probability, the probability on the other two rows is also
$\frac{1}{3}$. This gives a regret of at least
\[
   a-\left(\frac{1}{3}a+\frac{1}{3}(0\cdotp05)+\frac{1}{3}(0\cdotp1)\right) =\frac{2}{3}a-0\cdotp05
   \geq\frac{2}{3}(0\cdotp85)-0\cdotp05 \approx0\cdotp517
\]
The analysis for $r_2\leq r_1,r_3$ where the column player plays
$\vc{s}_2$ is similar.

Assume $r_3\leq r_1,r_2$ and the column player plays $\vc{s}_3$. The
best response to $\vc{s}_3$ has a payoff of at least $0\cdotp9$ and row 1 and 2
combined can have a payoff of at most $0\cdotp1$. This gives a regret of at least
\[
   a-\left(\frac{1}{3}a+\frac{1}{3}(0\cdotp1)\right) = \frac{2}{3}a-\frac{1}{30}
   \geq\frac{2}{3}(0\cdotp9)-\frac{1}{30} \approx0\cdotp567
\]
So regardless the strategy of the row player, the regret of the row
player is always larger than $0\cdotp501$ when the commitment of the other
player is at most $0\cdotp05$.

\bigskip\noindent
{\bf Case~\ref{case3}: Neither player has low commitment}

\smallskip\noindent
Suppose both players have commitment $\tau^r,\tau^c\geq 0\cdotp05$.
Consider the following set of payoff matrices for the column player:
$C^1,\ldots,C^n$ where $C^\ell$ has a payoff of 1 for every entry in the
$\ell$-th column and a 0 elsewhere:
\[
\forall i,j:\quad C_{ij}^\ell = 1 {\rm ~if~} j=\ell;~ 0 {\rm ~otherwise}
\]
To achieve a $0\cdotp501$\NNE, when the column player has payoff matrix $C^\ell$,
the column player should assign at least $0\cdotp499$ to column $\ell$.

The construction of the payoff matrix $R$ of the row player will depend
on the centre strategy $\vc{c}^r$ of the row player.
Take the $(n-\sqrt{n})$ rows of $R$ which have the highest values $\vc{c}^r[i]$,
where $\vc{c}^r[i]$ is the $i$-th entry of $\vc{c}^r$.
We construct matrix $R$ for which these rows are all zero.
For the construction of the remaining $\sqrt{n}$ rows of $R$ we consider
$\vc{c}^c$, the centre distribution of the column player.
We select the $(n-\sqrt{n})$ columns $j$
of $R$ having the highest values $\vc{c}^c[j]$. If row $i$ is one
of the rows with one of the $\sqrt{n}$ smallest entries for
$\vc{c}^r$ and column $j$ is a column with one of the
$(n-\sqrt{n})$ highest entries for $\vc{c}^c$, then we set $R_{ij}=1$.
The payoff entries in $R$ that are still undefined can be seen as a
$(\sqrt{n}\times\sqrt{n})$-sub-matrix.

This submatrix will contain a submatrix $M_{n'}$ as in Definition~\ref{def:matrix},
where (for $k'=\lfloor\sqrt{n'}\rfloor$) $\binom{n'}{k'} = \sqrt{n}$.
The extra columns of the submatrix have all their payoffs set to 1.
The entire matrix $R$ now satisfies the conditions of Corollary~\ref{cor:key}.
Let $S$ be the set of non-zero rows;
by construction $\sum_{i\in S}\vc{c}^r[i]\leq\frac{1}{\sqrt{n}}$.
Since $d(\vc{x},\vc{c}^r)<0\cdotp95$, we have $\sum_{i\in S}\vc{x}[i]<0\cdotp05+\frac{1}{\sqrt{n}}$,
so $\sum_{i\not\in S}\vc{x}[i]>0\cdotp05-\frac{1}{\sqrt{n}}$.
Similarly $\vc{y}$ has measure $>0\cdotp05-\frac{1}{\sqrt{n}}$ on columns
not intersecting $M_{n'}$.

The values of $p_r$ and $p_c$ in Corollary~\ref{cor:key} are $0\cdotp05-\frac{1}{\sqrt{n}}$,
and the value of $p$ is $0\cdotp499$, so we get a regret of at least
$0\cdotp499-O(\frac{1}{\sqrt{n}})+(0\cdotp05-O(\frac{1}{\sqrt{n}}))^2$
$=0\cdotp5015-O(\frac{1}{\sqrt{n}})$.
\end{proof}

\section{One-way Communication}\label{sec:oneway}

We noted in Section~\ref{sec:approxalgos} that $\epsilon=\frac{1}{2}$ can be
achieved with one-way communication,
by a simple implementation of the DMP-algorithm, using logarithmic communication.
The following result gives a matching lower bound of $\frac{1}{2}$.
It thus also furnishes a slightly simpler lower-bound result for the communication-free
setting of the previous section, but of course the lower bound itself is necessarily weaker.

\begin{theorem}\label{thm:oneway}
It is impossible to guarantee to find an $\epsilon$-Nash equilibrium,
for any constant $\epsilon<\frac{1}{2}$, with unlimited one-way communication.
\end{theorem}

\begin{proof}
We consider games $G=(R,C)$, where $R$ and $C$ are payoff matrices with
dimensions $\binom{n}{k}\times n$, with $k\approx\sqrt{n}$. Consider the following
set of column player payoff matrices $C^1,\ldots,C^n$, where $C^\ell$ has a
payoff of 1 for every entry in the $\ell$-th column and a 0 elsewhere:
\[
\forall i,j:\quad C_{ij}^\ell = 1 {\rm ~if~} j=\ell;~ 0 {\rm ~otherwise}
\]
The row player has matrix $R=M_n$ with $M_n$ as in Definition~\ref{def:matrix}.

Let $\vc{x}$ be the strategy of the row player, resulting from matrix $R$.
Let $\vc{y}_\ell$ be the strategy of the column player resulting from matrices
$R$ and $C^\ell$; note that with unlimited one-way communication we can
assume that the row player communicates all of $R$ (and indeed, $\vc{x}$)
to the column player.

We will show that for this class of games, one cannot do better than a
$(\frac{1}{2}-o(\frac{1}{\sqrt{n}}))$\NNE.

We search for a lower bound of $\frac{1}{2}-z$,
and we identify that a value of $z$ of $\frac{1}{\sqrt{n}}$ applies.

First observe that a best response for the column player having matrix $C^\ell$ is
$\vc{e}_\ell$, the pure strategy of column $\ell$. Column $\ell$ has payoff 1 and other
columns have payoff 0. So to reach a $(\frac{1}{2}-z)$\NNE,
$\vc{y}_\ell$ must allocate a probability at least $(\frac{1}{2}+z)$ to column $\ell$.

So, an arbitrary column $\ell$ can be required to have probability at least
$\frac{1}{2}-z$. Lemma~\ref{lem:key} says that the row player's regret is at least
$\frac{1}{2}-z-O(\frac{1}{\sqrt{n}})$.
Put $z=\frac{1}{\sqrt{n}}$ and we find that for all $\vc{x}$, $\ell$ may be chosen
such that in order for the column player to have regret less than $\frac{1}{2}-O(\frac{1}{\sqrt{n}})$,
the row player must have regret at least $\frac{1}{2}-O(\frac{1}{\sqrt{n}})$.
\end{proof}

\begin{theorem}
It is impossible to guarantee to find an $\epsilon$-well-supported Nash equilibrium,
for any constant $\epsilon<1$, with unlimited one-way communication.
\end{theorem}

\begin{proof}
To prove this theorem we will only have to look at $2\times 2$ games.
The row player has the identity matrix and the column player has one of two different column matrices. Communication is only allowed from the row player to the column player.
\[
	R   = \left( \begin{array}{cc}1&0\\0&1\\ \end{array}\right)\quad \text{and}\quad
	C^1 = \left( \begin{array}{cc}1&0\\1&0\\ \end{array}\right)\quad \text{or} \quad 
	C^2 = \left( \begin{array}{cc}0&1\\0&1\\ \end{array}\right)
\]
In any $\epsilon$-\WSNE\ for $\epsilon<1$, the column player must play pure
strategy column $j$, given payoff matrix $C^j$. That is necessary regardless of the
information she receives from the row player.

No communication is allowed from the column player to the row player, so the row
player's strategy is determined by matrix $R$. Let $f^r(R)$ be the row player's
strategy. If $f^r(R)$ allocates positive probability to row $i$, then we fail
to have a $\epsilon$-\WSNE\ (for any $\epsilon<1$) when the column player has
matrix $C^{3-i}$, since when that happens, row $i$ pays the row player $0$
while the other row pays $1$.
\end{proof}

\section{Communication-bounded algorithms}\label{sec:commbounded}

This section present the main positive results, algorithms that compute
approximate Nash equilibria that are limited to polylogarithmic communication.
Section~\ref{sec:epsilonnash} gives the main result for $\epsilon$-Nash
equilibria, and Section~\ref{sec:wsne} gives a variation of the algorithm
that compute $\epsilon$-\WSNE\ for $\epsilon\approx0\cdotp732$.

\subsection{A $0\cdotp438$\NNE\ procedure with limited communication}\label{sec:0.438}
\label{sec:epsilonnash}

This section provides a $0\cdotp438$\NNE\ procedure where the amount of communication
between the players is polylogarithmic in $n$.
We present the algorithm as an $\alpha$\NNE\ procedure first and then optimize $\alpha$.
At various points the algorithm uses the operation of communicating a mixed strategy
(a probability distribution over $[n]$) from one player to the other;
the details of this operation are
given in Section~\ref{sec:cms}. The general idea is to communicate a sample of
size $O(\log n)$ from the distribution and argue that the corresponding
empirical distribution is a good enough estimate for our purposes.

First the row player finds a \NE\ for the zero-sum game $(R,-R)$ and
the column player computes a \NE\ for the zero-sum game $(-C,C)$. Since
both games are zero-sum, we know that the payoff values for their
Nash equilibria will be unique. Both players compare this payoff value with
$\alpha$. We distinguish two cases,
\begin{enumerate}\noitemsep
\item neither player can ensure himself a payoff more than $\alpha$, or
\item at least one of the players can ensure a payoff more than $\alpha$.
\end{enumerate}
With $O(1)$ communication, the case that holds can be identified.

\bigskip\noindent
{\bf Case~1: the value of both zero-sum games is $\leq\alpha$ to each player}

\smallskip\noindent
The row player finds a strategy pair $(\vc{x}^*_r, \vc{y}^*_r)$ as solution to
$(R,-R)$, while the column player finds a strategy pair $(\vc{x}^*_c, \vc{y}^*_c)$
as solution to $(-C,C)$. The row player communicates $\vc{y}^*_r$ to the column
player (as described in Section~\ref{sec:cms}) and the
column player sends $\vc{x}^*_c$ to the row player.
They now play the game $(R,C)$ using strategy pair $(\vc{x}^*_c, \vc{y}^*_r)$.
Since $\vc{y}^*_r$ is a \NE\ strategy
in the zero-sum game $(R,-R)$ and the row player still plays with payoff
matrix $R$, by the minimax theorem, the row player has no strategy that
can give him a payoff of $\alpha$ or higher. The row player has a best
response with a value of at most $\alpha$, so his regret is also at most $\alpha$.
The strategy $\vc{x}^*_c$ was a \NE\ strategy in the zero-sum game
$(-C,C)$ and the column player still has payoff matrix $C$. So we can
use the same argument for the column player to claim that when the row
player plays strategy $\vc{x}^*_c$, the column player has regret at most
$\alpha$. So, we have a $\alpha$\NNE.
This concludes Case 1.

\bigskip\noindent
{\bf Case~2: one or both players can guarantee a payoff $>\alpha$}

\smallskip\noindent
If at least one of the players has a value of more than $\alpha$ for his
zero-sum game, he can get a payoff of more than $\alpha$ if he plays this
strategy, regardless the strategy of the other player. Assume w.l.o.g.\
that it is the row player who has a payoff greater than $\alpha$ in his
zero-sum game. He communicates this strategy $\vc{x}^*_r$ to
the column player (again, as described in Section~\ref{sec:cms}).
The column player identifies a pure strategy best
response $\vc{e}_j$ to $\vc{x}^*_r$ and communicates
$\vc{e}_j$ to the row player (using $\log n$ bits).

At this point in the algorithm we have the strategy pair
$(\vc{x}^*_r,\vc{e}_j)$. The column player has a best response
strategy, so at this point his regret is 0. The row player's strategy
$\vc{x}^*_r$ is paying him more than $\alpha$.
Let $\beta\leq 1$ be the value of his best response to $\vc{e}_j$.
So at this point the row player has a regret of at most  $\beta-\alpha$.
We next deal with the possibility that $\beta-\alpha>\alpha$.

At this stage the column player has regret 0 while we are only looking
for regret to be bounded by $\alpha$; meanwhile the row player has a strategy that might not
be good enough for an $\alpha$\NNE. To change this, we use a method
used in~\cite{Chen2009} (Lemma 3.2), which allows the row player to shift some of
his probability to his best response to $\vc{e}_j$. By shifting some of
his probability, it could be that $\vc{e}_j$ no longer is a best
response strategy for the column player. This is acceptable, as long as the
column player's regret while playing $\vc{e}_j$ is at most $\alpha$.
Suppose the row player shifts $\frac{1}{2}\alpha$ of his probability to a best
response strategy.
The payoff the column player gets with $\vc{e}_j$ could be
$\frac{1}{2}\alpha$ lower because of this move. The payoff of some other
column(s) could go as much as $\frac{1}{2}\alpha$ higher because of this
shift. The strategy $\vc{e}_j$ had regret 0, so by the shift of
$\frac{1}{2}\alpha$ of the row player's probability, the regret of the
column player is at most $\frac{1}{2}\alpha+\frac{1}{2}\alpha=\alpha$,
which constitutes an $\alpha$\NNE, for the column player.

The row player is allowed to change the allocation of $\frac{1}{2}\alpha$
of his probability that was allocated to strategies having the lowest payoff.
The remainder of his
probability, $1-\frac{1}{2}\alpha$, had already at least an average payoff of $\alpha$.
The probability is shifted to his best response with a value of $\beta$,
with $\alpha\leq\beta\leq1$. The following inequality is a sufficient
condition for the row player's regret to be at most $\alpha$:
\[
\left(1-\frac{1}{2}\alpha\right)\alpha+\frac{1}{2}\alpha\beta\geq\beta-\alpha \ , \quad 0\leq\alpha\leq\beta\leq1
\]

The solutions to this inequality are
\[
\begin{tabular}{rl}
$0<\alpha\leq\frac{1}{2}(5-\sqrt{17})$ & $\quad \alpha\leq \beta\leq\frac{\alpha^2-4\alpha}{\alpha-2}$ \\
$\frac{1}{2}(5-\sqrt{17})<\alpha<1$ &    $\quad \alpha\leq \beta\leq1$ \\
$\alpha=0\quad \beta=0$ &  $\quad \alpha=1\quad \beta=1$
\end{tabular}
\]
where it holds that if $\alpha=\frac{1}{2}(5-\sqrt{17})$ then
$f(\alpha)=\frac{\alpha^2-4\alpha}{\alpha-2}=1$ and for $0\leq\alpha\leq1$
this function is monotone increasing. This procedure will give an
$\alpha$\NNE, so $\alpha$ should be as low as possible.
Next to this it should also hold for every $\beta$ with $\alpha\leq\beta\leq1$.
The lowest $\alpha$ such that this condition hold is when $f(\alpha)=1$,
thus $\alpha=\frac{1}{2}(5-\sqrt{17})\approx 0\cdotp438$.

So if the row player rearranges $\frac{1}{2}\cdot0\cdotp438=0\cdotp219$ of his
probability to his best response row, both players have a strategy
that guarantees them a $0\cdotp438$\NNE.

\subsection{Communicating Mixed strategies}\label{sec:cms}

We describe how to communicate an approximation of the
mixed strategies that are computed, using $O(\log^2 n)$ bits.
We ultimately obtain an $\epsilon$ of $0\cdotp438+\delta$, for any $\delta>0$.

We first look at the case where one of the players, assume w.l.o.g.\ the
row player, has a payoff higher than $\alpha$ in the \NE\ of his
zero-sum game $(R,-R)$. The column player plays a pure best response to the
strategy of the row player, regardless of the support of the strategy of
the row player. So we mainly consider the row player.

The zero-sum game $(R,-R)$ gives a strategy pair $(\vc{x}^*,\vc{y}^*)$.
Fix $k=\frac{\ln n}{\delta^2}$ and form a multiset $A$ by sampling $k$
times from the set of pure strategies of the row player, independently
at random according to the distribution $\vc{x}^*$. Let $\vc{x}'$ be the
mixed strategy for the row player with a probability of $\frac{1}{k}$
for every member of $A$. We want the distribution $\vc{x}'$ to have a
payoff close to the payoff of $\vc{x}^*$. This corresponds to the
following event:
$$\phi=\{((\vc{x}')^TR\vc{y}^*)-((\vc{x}^*)^TR\vc{y}^*)<-\delta\}$$

As noted in~\cite{Lipton2003} the expression
$((\vc{x}')^TR\vc{y}^*)$ is essentially a sum of $k$ independent random
variables each of expected value $((\vc{x}^*)^TR\vc{y}^*)$, where every
random variable has a value between 0 and 1. This means we can bound the
probability that $\phi$ does not hold, which we will call $\bar{\phi}$. When
we apply a standard tail inequality \cite{Hoeffding63} to bound the
probability of $\bar{\phi}$, we get: $$\Pr[\bar{\phi}]\leq e^{-2k\delta^2}$$

With $k=\frac{\ln n}{\delta^2}$, this gives
$\Pr[\bar{\phi}]\leq\frac{1}{n^2}$ and $\Pr[\phi]\geq1-\frac{1}{n^2}$. If
$\vc{x}'$ does not give payoffs close enough to $\vc{x}^*$, we sample again.

The strategy $\vc{x}'$ has a guaranteed payoff of
$0\cdotp438+\delta-\delta=0\cdotp438$. This strategy is communicated to the column
player. The support of this strategy is logarithmic and all
probabilities are rational (multiples of $\frac{1}{k}$). Communication
of one pure strategy has a \CC\ of $O(\log n)$. This will give a \CC\
for $\vc{x}'$ of $O(\log^2 n)$.

The column player computes a pure strategy best response to $\vc{x}'$
and communicates this strategy in $O(\log n)$ to the row player. The
strategy of the row player might not yet lead to a $0\cdotp438$\NNE, his
payoff could be too low. As we have seen before, if the row player
redistributes at most $0\cdotp219$ of his probability, he is guaranteed to have
a strategy that leads to a $0\cdotp438$\NNE.

This change in strategy of the row player can decrease the payoff of the
column player by as much as $0\cdotp219$ and increase another pure strategy
by as much as $0\cdotp219$. His strategy was a best response, a 0\NNE, and
the improvement to another pure strategy is maximal $0\cdotp219+0\cdotp219=0\cdotp438$,
this leads to a $0\cdotp438$\NNE.

In the alternative case, where both players have a low ($<\alpha$)
payoff in their zero-sum games, the technique is essentially the same:
each player samples $k$ times from the opposing distribution, checks
that it limits his own payoff to at most $\alpha+\delta$, re-samples
as necessary, and communicates the $k$-sample.

\subsection{A $0\cdot732$-\WSNE\ procedure with limited communication}
\label{sec:wsne}

We give a variant of the algorithm of the previous section, that produces an
$\epsilon$-well-supported \NE\ for $\epsilon=\sqrt{3}-1$.
Like the previous algorithm, we will first search for
an $\alpha$\NNE\ and later find the optimal value for $\alpha$.

The algorithm starts in the same way as in Section~\ref{sec:epsilonnash}
with both players computing the \NE\ of zero-sum
games. The row player solves the zero-sum game $(R,-R)$ and the column
player solves $(-C,C)$. The two cases that arise are also the same; case
1 proceeds as in Section~\ref{sec:epsilonnash} while Case 2 requires a
variation to the algorithm.

\bigskip\noindent
{\bf Case~1: the value of both zero-sum games is $\leq\alpha$ to each player}

\smallskip\noindent
First consider the case where both players have a \NE\ with value
smaller than $\alpha$. The row player has a strategy pair
$(\vc{x}^*_r, \vc{y}^*_r)$ and the column player a
strategy pair $(\vc{x}^*_c, \vc{y}^*_c)$. The row
player communicates $\vc{y}^*_r$ to the column player and the
column player sends $\vc{x}^*_c$ to the row player. They will
now play the game with the strategy pair $(\vc{x}^*_c,
\vc{y}^*_r)$. If they play according to these strategies, then no pure strategy
yields a payoff of $\alpha$ or more, so note that the strategy profile is an
$\alpha$-\WSNE.

\bigskip\noindent
{\bf Case~2: one or both players can guarantee a payoff $>\alpha$}

\smallskip\noindent
Suppose that a player, assume w.l.o.g.\ the row player, has a payoff
more than $\alpha$ in the \NE\ of his zero-sum game $(R,-R)$.
Let the row player communicate this strategy $\vc{x}^*_r$ to the column player.
The column player computes a pure strategy best response $\vc{e}_j$ to $\vc{x}^*_r$
and communicates this strategy to the row player.
Because the row player had a payoff of at least $\alpha$ in the game
$(R,-R)$, he also has a payoff of at least $\alpha$ against $\vc{e}_j$.

At this point in the algorithm we have a strategy pair
$(\vc{x}^*_r,\vc{e}_j)$. The strategy of the column player is a
best response to $\vc{x}^*_r$, so his strategy has regret 0.
We have no guarantee on the performance of the row player's strategy, in the
context of a \WSNE.

As in the previous algorithm we allow the row player to shift some of his
probability to his best response to $\vc{e}_j$. Note that if we
shift $\frac{1}{2}\alpha$ of the probability of the row player, this
ensures the column player's payoffs vary by at most $\alpha$.

Let the best response of the row player to $\vc{e}_j$ have value $\beta\geq\alpha$.
The row player's payoff is a random variable $x$ that takes values
in $[0,1]$ with expectation $E(x)\geq\alpha$, since $\vc{x}^*_r$
is the security strategy for payoff matrix $R$. The maximum value $x$ can
take is $\beta$. The algorithm takes all strategies for which the row player's
payoff is less than $\beta-\alpha$, and replaces any probability allocated to
them by $\vc{x}^*_r$, to any strategy whose payoff is at least $\beta-\alpha$,
thus satisfying the conditions for the row player to also have an
$\alpha$-\WSNE.

We upper bound the probability $\Pr(x\leq\beta-\alpha)$ as follows.
Subject to $E(x)\geq\alpha$ and $\max(x)=\beta$, this is maximised when
$x$ takes values $\beta$ or $\beta-\alpha$.
Let $p=\Pr(x\leq\beta-\alpha)$. Then
\[
E(x)<p(\beta-\alpha)+(1-p)\beta=-\alpha p+\beta.
\]
We have $E(x)\geq\alpha$. Plugging that into the above,
\[
\alpha\leq-\alpha p+\beta,~~~~~~{\rm i.e.}~~p\leq\frac{\beta-\alpha}{\alpha}.
\]
To ensure that the amount of probability shifted is at most $p$, is
suffices to let $\frac{1}{2}\alpha < \frac{\beta-\alpha}{\alpha}$,
i.e. $\alpha^2+2\alpha-2\beta\geq 0$.
This is satisfied by $\alpha=-1+\sqrt{1+2\beta}$, so that the worst case value
of $\beta$ is 1, resulting in the claimed value of $\sqrt{3}-1\approx0\cdotp732$.

\section{Conclusions}\label{sec:conclusions}

Our results raise some open problems, such as how good an approximation
should be achievable in the communication-free setting, and how well we can
do in the setting of limited (two-way) communication.
Our communication-bounded algorithms are also based on algorithms that compute approximate
equilibria {\em in polynomial time}, and it would be very interesting if further upper bounds on the
communication complexity could be obtained for algorithms whose computational time
was not known to be polynomial.
Pastink~\cite{Pastink12} considers some related topics, including the communication
required for approximate equilibria of games of fixed size.
It may be that future work should address the issue of
communication protocols where the players have an incentive to report
their information truthfully.

We believe that the communication-limited algorithm for $0\cdotp438$-approximate Nash
equilibria is significant, also the $0\cdotp501$ lower bound in the communication-free
setting, since in the context of searching for $\epsilon$-approximate Nash
equilibria, $\epsilon=0\cdotp5$ frequently seems to arise as a limit on what is achievable.
For example, if we search for approximate equilibria of constant support, the
DMP-algorithm~\cite{Daskalakis2009} achieves this for $\epsilon=0\cdotp5$, however,
Feder et al.~\cite{feder2007} show that for $\epsilon<0\cdotp5$, the support size may need to
be logarithmic in $n$. (The corresponding logarithmic upper bound on the support size
that may be needed, is due to~\cite{Lipton2003}.) In a similar way, while Fictitious Play is known to guarantee
to find $\epsilon$-approximate equilibria for $\epsilon$ approaching $0\cdotp5$~\cite{Con2009}, it has
also been established that $\epsilon=0\cdotp5$ is, in the worst case, a lower bound on the approximation
quality attainable~\cite{GSSV2012}. And, as we find in Theorem~\ref{thm:oneway},
$0\cdotp5$ is also the best approximation that can be guaranteed when there is a
restriction to one-way communication. Finally, Fearnley et al.~\cite{FGGS12}
show that to find $\epsilon$-Nash equilibria with $\epsilon\geq\frac{1}{2}$, a
strictly smaller fraction of the payoffs of the game need to be checked, than
is needed for certain smaller positive values of $\epsilon$.

{\it Acknowledgements:}
The first author thanks Sergiu Hart for useful
discussions during the iAGT workshop in May 2011.

\end{document}